\documentclass{llncs}

\usepackage{epsfig}
\usepackage{algorithmic}
\usepackage{algorithm} 
\usepackage{graphicx}
\usepackage{amssymb}
\usepackage{fullpage}
\usepackage{color}
\usepackage{amsmath}
\usepackage{mathtools}
\usepackage{cite}

\usepackage{versions}
\excludeversion{INUTILE}
\includeversion{LONG}\excludeversion{SHORT}\includeversion{VLONG}
\markversion{TODO}

\begin{document}
\hyphenation{pro-blems sol-ving complexi-ty dia-gram poly-gonal analy-sis ins-tance des-cribed}

\pagestyle{headings}  
\addtocmark{        Input Order Adaptivity in Computational Geometry            }

\mainmatter              
\title{             Refining the Analysis of \\Divide and Conquer:\\ How and When            }
\titlerunning{      Refining the Analysis of \\Divide and Conquer:\\ How and When            } 
\author{
  J\'er\'emy Barbay\inst{1} 
  \and 
  Carlos Ochoa\inst{1}\thanks{Corresponding author.}
  \and
  Pablo P\'erez-Lantero\inst{2}
}

\institute{
  Departamento de Ciencias de la Computaci\'on, Universidad de Chile, Chile\\ 
  \email{jeremy@barbay.cl, cochoa@dcc.uchile.cl.}
  \and
  Escuela de Ingenier\'ia Civil Inform\'atica, Universidad de Valpara\'iso, Chile.\\ 
  \email{pablo.perez@uv.cl.}
}

\maketitle              

\begin{abstract}
Divide-and-conquer is a central paradigm for the design of algorithms, through which some fundamental computational problems, such as sorting arrays and computing convex hulls, are solved in optimal time within $\Theta(n\log{n})$ in the worst case over instances of size $n$.  A finer analysis of those problems yields complexities within $O(n(1 + \mathcal{H}(n_1, \dots, n_k))) \subseteq O(n(1{+}\log{k})) \subseteq O(n\log{n})$ in the worst case over all instances of size $n$ composed of $k$ ``easy'' fragments of respective sizes $n_1, \dots, n_k$ summing to $n$, where the entropy function $\mathcal{H}(n_1, \dots, n_k) = \sum_{i=1}^k{\frac{n_i}{n}}\log{\frac{n}{n_i}}$ measures the ``difficulty'' of the instance. We consider whether such refined analysis can be applied to other algorithms based on divide-and-conquer, such as polynomial multiplication, input-order adaptive computation of convex hulls in 2D and 3D, and computation of Delaunay triangulations.
\end{abstract}

\begin{center}
  \begin{minipage}{.9\textwidth}
  \end{minipage}
\end{center}

\section{Introduction}

The divide-and-conquer paradigm is used to solve central computational problems such as {\sc{Sorting}} arrays~\cite{1973-BOOK-TheArtOfComputerProgrammingVol3-Knuth}, computing {\sc{Convex Hulls}}~\cite{1986-JCom-TheUltimatePlanarConvexHullAlgorithm-KirkpatrickSeidel}, {\sc{Delaunay Triangulations}} and {\sc{Voronoi Diagrams}}~\cite{1985-ACMT-PrimitiveForTheManipulationOfGeneralSubdivisionAndTheComputationOfVoronoiDiagrams-GuibasStolfi} of points in the plane and in higher dimensions, {\sc{Matrix Multiplication}}~\cite{1969-NM-GaussianEliminationIsNotOptimal-Strassen}, {\sc{Polynomial Multiplication}}~\cite{1988-BOOK-NumericalRecipesInC-PressFlanneryTeukolskyVetterling}, {\sc{Half-Plane Intersection}}~\cite{1985-BOOK-ComputationalGeometryAnIntroduction-PreparataShamos}, among others. For some of these problems, this paradigm yields an optimal running time within $\Theta(n\log{n})$ in the worst case over instances of $n$ elements. 

An adaptive analysis of slight variants of some of these algorithms yields improved running times on large classes of instances. Those results can be refined~\cite{1976-JComp-SortingAndSearchingInMultisets-MunroSpira, 2013-TCS-OnCompressingPermutationsAndAdaptiveSorting-BarbayNavarro, 2009-FOCS-InstanceOptimalGeometricAlgorithms-AfshaniBarbayChan} up to complexities within $O(n(1 + \mathcal{H}(n_1, \dots, n_k))) \subseteq O(n(1{+}\log{k})) \subseteq O(n\log{n})$ in the worst case over all instances of size $n$ composed of $k$ ``easy'' fragments of respective sizes $n_1, \dots, n_k$ summing to $n$, where the entropy function $\mathcal{H}(n_1, \dots, n_k) = \sum_{i=1}^k{\frac{n_i}{n}}\log{\frac{n}{n_i}}$ measures the ``difficulty'' of the instance.
We describe here two examples for the problem of {\sc{Sorting}} arrays (see Section~\ref{sec:classification} for more).

\begin{itemize}
\item Munro and Spira~\cite{1976-JComp-SortingAndSearchingInMultisets-MunroSpira} showed that the algorithm {\tt{MergeSort}} can be adapted to sort a multiset $S$ of $n$ elements in time within $O(n(1 + \mathcal{H}(m_1, \dots, m_\sigma))) \subseteq O(n(1{+}\log{\sigma})) \subseteq O(n\log{n})$, where $\sigma$ is the number of distinct elements in $S$ and $m_1, \dots, m_\sigma$ (such that $\sum_{i=1}^\sigma {m_i}=n$) are the $\rho$ multiplicities of the distinct elements in $S$, respectively.

\item Taking advantage of the input order, Knuth~\cite{1973-BOOK-TheArtOfComputerProgrammingVol3-Knuth} considered sequences formed by \emph{runs} i.e., contiguous increasing subsequences, and described an algorithm sorting such sequences in time within $O(n(1+\log\rho))\subseteq O(n\lg n)$, where $\rho$ is the number of \emph{runs} in the sequence.  Barbay and Navarro~\cite{2013-TCS-OnCompressingPermutationsAndAdaptiveSorting-BarbayNavarro} improved the analysis of this algorithm in time within $O(n(1+\mathcal{H}(r_1, \dots, r_{\rho}))) \subseteq O(n(1{+}\log{\rho})) \subseteq O(n\log{n})$, where $r_1, \dots, r_{\rho}$ (such that $\sum_{i=1}^\rho {r_i}=n$) are the sizes of these \emph{runs}.
\end{itemize}

Similar analysis techniques have been applied to some other problems, though only partially. Following an example for the problem of computing the {\sc{Convex Hull}} in the plane (see Section~\ref{sec:comp-conv-hulls} for more): Levcopoulos et al.~\cite{2002-SWAT-AdaptiveAlgorithmsForConstructingConvexHullsAndTriangulationsOfPolygonalChains-LevcopoulosLingasMitchell} described an adaptive algorithm for computing the {\sc{Convex Hull}} of a polygonal chain. The algorithm takes advantage of the minimum number $\kappa$ of simple subchains into which the polygonal chain can be partitioned. They showed that the time complexity of the algorithm is within $O(n(1+\log{\kappa})) \subseteq O(n\log{n})$.

\begin{INUTILE}
The function ${\cal H}(n_1,\ldots,n_h)$ is the minimum entropy of the distribution of the points into a certificate of the instance.
\end{INUTILE}

\paragraph{Hypothesis.}
Which similar refinements can be applied to which divide-and-conquer algorithms, if any, and to what depth?

\paragraph{Our Results.}
In Section~\ref{sec:classification}, we list and classify previous refined analyses between those that are \emph{Structure Based} and those \emph{Input-Order Based}.
In Section~\ref{sec:polynomial}, we describe a refined analysis of the principal step in the algorithm for {\sc{Polynomial Multiplication}} using the \emph{Fast Fourier Transformation}.
In Sections~\ref{sec:comp-conv-hulls} and \ref{sec:monotone}, we describe two distinct refined analyses for problems in computational geometry, which yield various optimal input-order adaptive results on the computation of \textsc{Convex Hulls}, \textsc{Delaunay Triangulations}, and \textsc{Voronoi Diagrams} in the plane. 
\begin{LONG}
In Section~\ref{sec:comp-conv-hulls}, we refine the analysis of Levcopoulos et al.~\cite{2002-SWAT-AdaptiveAlgorithmsForConstructingConvexHullsAndTriangulationsOfPolygonalChains-LevcopoulosLingasMitchell}'s algorithm for computing the {\sc{Convex Hull}} of polygonal chains in time within $O(n(1 + {\mathcal{H}(n_1, \dots, n_\kappa)})) \subseteq O(n(1{+}\log{\kappa}))) \subseteq O(n\log{n})$, where $n_1, \dots, n_\kappa$ are the lengths of the subchains of a partition of a polygonal chain of $n$ points into the minimum number $\kappa$ of simple subchains. 
In Section~\ref{sec:monotone}, we describe a refined analysis of the computation of {\sc{Voronoi Diagrams}} and {\sc{Delaunay Triangulations}} for sequences $S$ formed by $n$ points, which yields a time complexity within $O(n(1+{\cal H}(v_1,\ldots,v_\mu))) \subseteq O(n(1{+}\log{\mu})) \subseteq O(n\log{n})$, where $v_1,\ldots,v_\mu$ are the sizes of the minimum number $\mu$ of monotone histograms in which $S$ can be cut, with respect to two fixed orthogonal directions\begin{VLONG} and show, as a corollary, an upper bound for computing the {\sc{Convex Hull}} of a sequence $S$ formed by $n$ points in time within $O(n(1+{\cal H}(v_1,\ldots,v_\mu))) \subseteq O(n(1{+}\log{\mu})) \subseteq O(n\log{n})$, where $v_1,\ldots,v_\mu$ are the sizes of the minimum number $\mu$ of monotone histograms in which $S$ can be cut, with respect to two fixed orthogonal directions\end{VLONG}.
\end{LONG}
In Section~\ref{sec:when}, we describe some more difficult applications of such refined analyses.
In Section~\ref{sec:discussion}, we discuss the possibility of designing algorithms which analyses combine \emph{Structure Based} and \emph{Input-Order Based} results synergically as opposed to running them in parallel.

\section{Classification of Results}
\label{sec:classification}

We review here some results on the refined analysis of algorithms for {\sc{Sorting}} arrays and for computing planar {\sc{Convex Hulls}}, {\sc{Delaunay Triangulations}} and {\sc{Voronoi Diagrams}}. We classify the various refined analysis between those focusing on the structure  of the instance (Section~\ref{sec:structure}) versus those focusing on the order in which the input is given (Section~\ref{sec:inputOrder}). There does not seem to be any example where both strategies are mixed: we discuss the potential for those in Section~\ref{sec:discussion}.

\subsection{Structure Based Results}
\label{sec:structure}

By ``Structure Based Results'' we mean algorithms taking advantage of the structure of the instance, for example, taking advantage of the frequencies of the values in a multiset or of the relative positions of the points in a set. Such results are known for {\sc{Sorting}} multisets and for computing {\sc{Convex Hulls}}.

\begin{LONG}
{\tt{MergeSort}} is a divide-and-conquer {\sc{Sorting}} algorithm in the comparison model. This algorithm relies on a linear time merge process, 
that combines two ordered sequences into a single ordered sequence.
\end{LONG}

Concerning the problem of {\sc{Sorting}}, Munro and Spira~\cite{1976-JComp-SortingAndSearchingInMultisets-MunroSpira} considered the task of {\sc{Sorting}} a multiset $S=\{x_1, \dots, x_n\}$ of $n$ real numbers with $\sigma$ distinct values, of multiplicities $m_1, \dots, m_\sigma$, respectively, so that $\sum_{i=1}^\sigma {m_i}=n$. They showed that adding counters to various classical algorithms (among which the divide-and-conquer based algorithm {\tt{MergeSort}}) yields a time complexity within $O(n(1+\mathcal{H}(m_1, \dots, m_\sigma))) \subseteq O(n(1{+}\log{\sigma})) \subseteq O(n\log{n})$ for {\sc{Sorting}} a multiset, where $\mathcal{H}(m_1, \dots, m_\sigma) = \sum_{i=1}^\sigma{\frac{m_i}{n}}\log{\frac{n}{m_i}}$ measures the entropy of the distribution of the multiplicities $\langle m_1, \dots, m_\sigma \rangle$. This result takes advantage of the frequencies of the values i.e., the structure of the instance.

\begin{LONG}
 Given a set $P$ of $n$ points, the \emph{Convex Hull} of $P$ is the smallest convex set containing $P$ (see Figure~\ref{fig:structures}).
\end{LONG}
Considering the problem of computing the {\sc{Convex Hull}}, Kirkpatrick and Seidel~\cite{1986-JCom-TheUltimatePlanarConvexHullAlgorithm-KirkpatrickSeidel} described an algorithm to compute the {\sc{Convex Hull}} of a set of $n$ planar points in time within $O(n(1+\log h))\subseteq O(n\log n)$, where $h$ is the number of vertices in the {\sc{Convex Hull}}. The algorithm relies on a variation of the divide-and-conquer paradigm, which they call the ``Marriage-Before-Conquest'' principle.
\begin{LONG}
For computing the upper hull, the algorithm finds a vertical line that divides the input point set into two approximately equal-size parts in liner time.  Next, it determines the edge of the upper hull that intersects this line in linear time.  It then eliminates the points that lie underneath this edge and finally applies the same procedure to the two sets of the remaining points on the left and right side of the vertical line.  A similar algorithm computes the lower hull.
\end{LONG}
Afshani et al.~\cite{2009-FOCS-InstanceOptimalGeometricAlgorithms-AfshaniBarbayChan} refined the complexity analysis of this algorithm to within $O(n(1+{\cal H}(n_1,\ldots,n_h)))\subseteq O(n(1{+}\log h)) \subseteq O(n\log{n})$, where $n_1, \dots, n_h$ are the sizes of a partition of the input, such that every element of the partition is a singleton or can be enclosed by a triangle whose interior is completely below the upper hull of the set, and ${\cal H}(n_1,\ldots,n_h)$ has the minimum possible value (minimum entropy of the distribution of the points into a certificate of the instance). This result takes advantage of the positions of the points i.e., the structure of the instance.

\subsection{Input-Order Based Results}
\label{sec:inputOrder}

By ``Input-Order Based Results'' we mean algorithms taking advantage of the order of the input, for example, taking advantage of the order of the values in a sequence of numbers or of the order in which the points are given in a polygonal chain. We describe only a sampling of such results on the problem of {\sc{Sorting}} permutations (see the survey from Moffat and Petersson~\cite{1992-ACJ-AnOverviewOfAdaptiveSorting-MoffatPetersson} for more), revisit some results on the computation of {\sc{Convex Hulls}} in the plane and 3D space, and show that those results in the plane are actually only ``Input-Order Based''. Those results for computing {\sc{Convex Hulls}} in 3D space show that no algorithm can take advantage of the position of the points i.e., structure based, in order to compute {\sc{Delaunay triangulations}} in the plane. 

Concerning the problem of {\sc{Sorting}}, Knuth~\cite{1973-BOOK-TheArtOfComputerProgrammingVol3-Knuth} described an adaptive sorting algorithm that takes advantage of permutations formed by sorted blocks called \emph{runs}, that is, subsequences of consecutive positions in the input with a positive gap between successive values, from beginning to end. He showed that the time complexity of this algorithm is within $O(n(1{+}\log\rho))\subseteq O(n\lg n)$, where $\rho$ is the number of \emph{runs} in the permutation\begin{SHORT}.\end{SHORT}
\begin{LONG}
(e.g. $(\emph{1,2,6,7,8,9},3,4,5)$ is composed of 2 such sorted blocks $(1,2,6,7,8,9)$ and $(3,4,5)$).
\end{LONG}
Barbay and Navarro~\cite{2013-TCS-OnCompressingPermutationsAndAdaptiveSorting-BarbayNavarro} refined the algorithm described by Knuth~\cite{1973-BOOK-TheArtOfComputerProgrammingVol3-Knuth}. They included in the analysis not only the number of \emph{runs} but also their sizes,
\begin{LONG}
 The main idea is to detect the runs first and then merge them pairwise, using a mergesort-like step.
The detection of ascending runs can be done in linear time by a scanning process identifying the positions $i$ in $\pi$ such that $\pi(i) > \pi(i+1)$.
Merging the two shortest runs at each step further reduces the number of comparisons, making the running time of the merging process adaptive to the entropy of the sequence of the lengths of the runs.
The merging process is then represented by a tree with the shape of a Huffman~\cite{1952-IRE-AMethodForTheInstructionOfMinimumRedundancyCodes-Huffman} tree, built from the distribution of the \emph{runs} sizes. They extend this result to mix ascending and descending runs,
\end{LONG}
showing that, if the permutation $\pi$ is formed by $\rho$ runs of sizes given by the vector $\langle r_1, \dots, r_{\rho} \rangle$, then $\pi$ can be sorted in time within $O(n(1+\mathcal{H}(r_1, \dots, r_{\rho}))) \subseteq O(n(1{+}\log{\rho})) \subseteq O(n\log{n})$. This result takes advantage of the order of the values in the input i.e., the input order.

Considering the computation of the {\sc{Convex Hull}} in the plane, Levcopoulos et al.~\cite{2002-SWAT-AdaptiveAlgorithmsForConstructingConvexHullsAndTriangulationsOfPolygonalChains-LevcopoulosLingasMitchell} described a divide-and-conquer algorithm for computing the {\sc{Convex Hull}} of a polygonal chain. The algorithm is based in the fact that the {\sc{Convex Hull}} of a simple chain can be computed in linear time, and that deciding whether a given chain is simple can be done in linear time.
\begin{INUTILE} simplicity of a chain can also be tested in linear time.\end{INUTILE}
They measured the complexity of this algorithm in terms of the minimum number of simple subchains $\kappa$ into which the chain can be cut.  They showed that the time complexity of this algorithm is within $O(n(1{+}\log{\kappa})) \subseteq O(n\log{n})$. We improve the analysis of this algorithm including not only the minimum number of simple subchain into which the polygonal chain can be partitioned but also their sizes (see Section~\ref{sec:comp-conv-hulls}).
This result takes advantage of the order in which the points are given i.e., the input order.

\begin{LONG}
\begin{figure}
\centering
\includegraphics[scale=1.2]{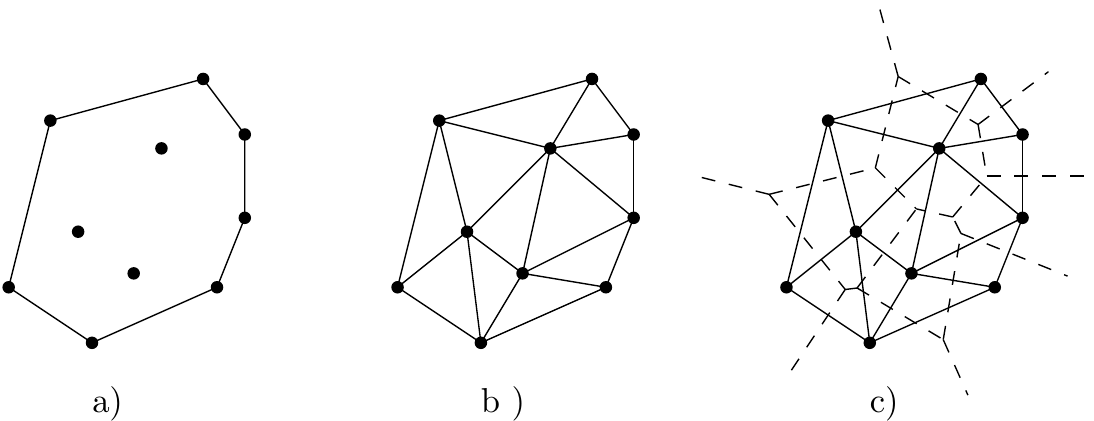}
\caption{A point set $P$. a) The convex hull of $P$, b) the Delaunay triangulation of $P$, and c) the Delaunay triangulation and the Voronoi diagram of $P$.}
\label{fig:structures}
\end{figure}
\end{LONG}

\begin{LONG}
Given a set $P$ of $n$ planar points, a \emph{triangulation} of $P$ is a subdivision of the convex hull of $P$ into triangles with vertex set the set $P$.
\end{LONG}
Concerning the computation of the {\sc{Convex Hull}} in 3D, a related concept is that of the {\sc{Delaunay Triangulation}} $DT(P)$ of a point set $P$ in the plane, a triangulation where \begin{LONG} for every edge $e$ there exists a disk $C$ with the following properties: (i) the endpoints of edge $e$ are on the boundary of $C$, and (ii) no other point of $P$ is in the interior of $C$: it is named after Boris Delaunay for his work on this topic in 1934. An equivalent definition is such that\end{LONG} no point in $P$ is inside the circumcircle of any triangle of $DT(P)$.  
Computing the \textsc{Delaunay Triangulation} is equivalent to computing its dual, called the \textsc{Voronoi Diagram}\begin{LONG}: each one can be constructed from the other in linear time~\cite{1985-BOOK-ComputationalGeometryAnIntroduction-PreparataShamos} (see Figure~\ref{fig:structures})\end{LONG}.
Computing the \textsc{Delaunay Triangulation} of a set of points in two dimensions reduces to computing the \textsc{Convex Hull} in three dimensions of the projections of those points on an hyperbolic plane.
\begin{LONG} The projection of $P$ onto the unit elliptic paraboloid $z=x^2+y^2$ yields a point set $P'$. The \textsc{Convex Hull} $CH(P')$ of $P'$ contains every point of the set. The downward-facing facet of $CH(P')$ are those whose normal vectors have a negative $z$-value. Projecting the edges of downward-facing facet in $CH(P')$ onto the plane yields the {\sc{Delaunay triangulation}} of $P$.
\end{LONG}
By showing tight bounds for input-order oblivious i.e., structure based, algorithms for computing \textsc{Convex Hulls} in three dimensions, Afshani et al.~\cite{2009-FOCS-InstanceOptimalGeometricAlgorithms-AfshaniBarbayChan} indirectly proved that no planar \textsc{Delaunay Triangulation} algorithm can take advantage of the position of the points.
\begin{LONG}
\begin{theorem}[Afshani et al.~\cite{2009-FOCS-InstanceOptimalGeometricAlgorithms-AfshaniBarbayChan}] \label{th:lowerBound} Consider a set of $n$ points in the plane.  For any algorithm $A$ computing the Delaunay triangulation in the algebraic decision tree model, $A$ performs in time within $\Omega(n\log n)$ on average on a random order of the points. This implies that there is an order of those points for which $A$ performs in time within $\Omega(n\log n)$.
\end{theorem}
\end{LONG}
We describe in Sections~\ref{sec:refin-analys-two} and \ref{sec:when} some algorithms taking advantage of the order of the input to compute \textsc{Delaunay Triangulations} and \textsc{Voronoi Diagrams}, among other desirable objects in computational geometry.        

\section{Refined Analysis: Three Examples}\label{sec:refin-analys-two}

We show in this section how to refine the analysis of the principal step in the algorithm for multiplying polynomials using the \emph{Fast Fourier Transformation}~\cite{1988-BOOK-NumericalRecipesInC-PressFlanneryTeukolskyVetterling}, how to refine the analysis of the algorithm from Levcopoulos et al.~\cite{2002-SWAT-AdaptiveAlgorithmsForConstructingConvexHullsAndTriangulationsOfPolygonalChains-LevcopoulosLingasMitchell} for the decomposition of a polygonal chain into simple sequences, and how to extend this analysis to the computation of {\sc{Voronoi diagrams}} and {\sc{Delaunay triangulations}} for another measure of difficulty based on monotone histograms.

\subsection{Polynomial Multiplication: Adaptivity to Zero-coefficients}
\label{sec:polynomial}

Given two polynomials $A = (a_0, \dots, a_{n-1})$ and $B = (b_0, \dots, b_{n-1})$ described by their coefficients, the polynomial multiplication problem is to compute the coefficients of the polynomial $C=A\cdot B$. The approach to multiplying polynomials using the \emph{Fast Fourier Transformation}~\cite{1988-BOOK-NumericalRecipesInC-PressFlanneryTeukolskyVetterling} can be divided into three steps: (i) evaluate the polynomials $A$ and $B$ in $2n$ values (the $(2n)$th roots of the unity); (ii) evaluate $C$ in these $2n$ values by multiplying the evaluations of $A$ and $B$; and (iii) obtain the coefficients of $C$ by interpolation using the values computed in the step (ii). The steps (i) and (iii) are accomplished by a divide-and-conquer algorithm for polynomial evaluation. We refine the analysis of the divide-and-conquer polynomial evaluation algorithm to take advantage of the number of zero-coefficients and of their relative positions in the vector of coefficients that describes the polynomial.

Given a polynomial $A$, the polynomial evaluation algorithm defines two polynomials, $A_{even}$ and $A_{odd}$, that consist of the even-indexed and odd-indexed coefficients of $A$, respectively. Hence, $A(x) = A_{even}(x^2) + x \times A_{odd}(x^2)$. If $x$ is one of the $(2n)$th roots of the unity, then $x^2$ is one of the $n$th roots of the unity. In order to evaluate the polynomial $A$ on each of the $(2n)$th roots of the unity, the recursive procedure divides $A$ into $A_{even}$ and $A_{odd}$, evaluates $A_{even}$ and $A_{odd}$ in the $n$th roots of the unity, and once these values are computed, evaluate $A$ on each of the $(2n)$th roots of the unity using the formula $A(x) = A_{even}(x^2) + x \times A_{odd}(x^2)$. The time complexity $T(n)$ of this algorithm follows the recurrence $T(n) \le T(\frac{n}{2}) + O(n)$, which yields a time complexity within $O(n\log{n})$. If at one step of the recursion call all the coefficients are zero, the algorithm finishes the computation at this branch (see Figure~\ref{fig:coefficients}).

\begin{figure}
\centering
\includegraphics[scale=0.9]{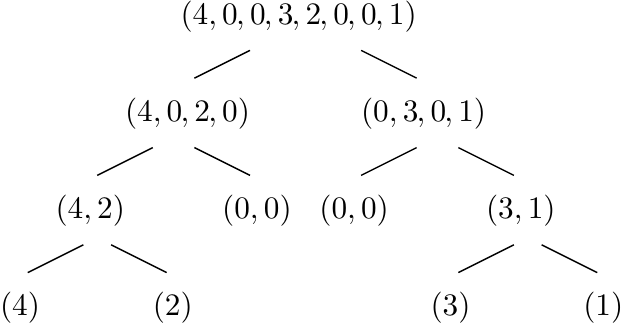}
\caption{Vector formed by the coefficients of the polynomial $A(x) = 4 + 3x^3 + 2x^4 + x^7$ and the recursion tree of the divide-and-conquer evaluation algorithm. At each step, the algorithm divides the coefficients of $A$ into $A_{even}$ and $A_{odd}$.}
\label{fig:coefficients}
\end{figure}

Given a polynomial $A=(a_0, \dots, a_{n-1})$, we define the equivalence relation $E$ between the positions of the zero-coefficients in $A$: two positions of zeros $p$ and $q$ are equivalent if and only if (i) there exists $k \in \mathbb{Z^+}$ such that $p \equiv q\ (\textrm{mod}\ 2^k)$; (ii) all the positions $r$ such that $r \equiv p\ (\textrm{mod}\ 2^k)$ are zeros; and (iii) there exist a position $t$ such that $t \equiv p \ (\textrm{mod}\ 2^{k-1})$ and the value in $t$ is different from zero. The idea of this equivalence relation is to group the positions in classes such that the positions in the same class form a vector where all the coefficients are zero in a node of the recursion tree. Let $\zeta$ and $\eta$ be the number of zeros in the vector formed by the coefficients of $A$ and the number of equivalence classes defined by $E$, respectively. Let $\langle n_1, \dots, n_{\eta} \rangle$ be the vector formed by the sizes of the equivalence classes defined by $E$. Then, $\sum^{\eta}_{i=1} n_i = \zeta$. The following theorem sets the refined analysis in function of $\zeta$ and the vector $\langle n_1, \dots, n_{\eta} \rangle$.

\begin{theorem}
The complexity of the evaluation algorithm for polynomials of $n$ coefficients, number of zero-coefficients $\zeta$ and vector $\langle n_1, \dots, n_{\eta}\rangle$ formed by the sizes of the equivalence classes defined by $E$ is within $O((n-\zeta)\log{n} + \sum^{\eta}_{i=1} n_i\log{\frac{n}{n_i}}) \subseteq O(n\log{n})$.  
\end{theorem}

In the following sections, we apply similar techniques to obtain optimal refined analysis for input-order adaptive algorithms computing {\sc{Convex Hulls}} and {\sc{Delaunay triangulations}} in the plane.

\subsection{Computing Convex Hulls: Adaptivity to Simple Subchains}
\label{sec:comp-conv-hulls}

A \emph{polygonal chain} is a curve specified by a sequence of points $p_1, p_2, \dots, p_n$. The curve itself consists of the line segments connecting the pairs of consecutive points. A polygonal chain $C$ is \emph{simple} if any two edges of $C$ that are not adjacent are disjoint, or if the intersection point is a vertex of $C$; and any two adjacent edges share only their common vertex. Melkman~\cite{1987-IPL-OnLineConstructionOfTheConvexHullOfASimplePolyline-Melkman} described an algorithm that computes the {\sc{Convex Hull}} of a simple polygonal chain in linear time, and Chazelle~\cite{1991-DCG-TriangulatingASimplePolygonInLinearTime-Chazelle} described an algorithm for testing whether a polygonal chain is simple in linear time.

Levcopoulos et al.~\cite{2002-SWAT-AdaptiveAlgorithmsForConstructingConvexHullsAndTriangulationsOfPolygonalChains-LevcopoulosLingasMitchell} combined these results to yield an adaptive divide-and-conquer algorithm for computing the {\sc{Convex Hull}} of polygonal chains. The algorithm tests if the chain $C$ is simple, using Chazelle~\cite{1991-DCG-TriangulatingASimplePolygonInLinearTime-Chazelle}'s algorithm: if the chain $C$ is simple, the algorithm computes the {\sc{Convex Hull}} of $C$ in linear time, using Melkman~\cite{1987-IPL-OnLineConstructionOfTheConvexHullOfASimplePolyline-Melkman}'s algorithm. Otherwise, if $C$ is not simple, the algorithm cuts $C$ into the subsequences $C'$ and $C''$, whose sizes differ at most in one; recurses on each of them; and merges the resulting {\sc{Convex Hulls}} using Preparata and Shamos's algorithm~\cite{1985-BOOK-ComputationalGeometryAnIntroduction-PreparataShamos}. They measured the complexity of this algorithm in terms of the minimum number of simple subchains $\kappa$ into which the chain $C$ can be cut. Let $t(n, \kappa)$ be the worst-case time complexity taken by this algorithm for an input chain of $n$ vertices that can be cut into $\kappa$ simple subchains.
They showed that $t(n, \kappa)$ satisfies the following recursion relation: $t(n, \kappa) \leq t(\lceil \frac{n}{2} \rceil, \kappa_1) + t(\lfloor \frac{n}{2} \rfloor, \kappa_2), \kappa_1 + \kappa_2 \leq \kappa + 1$. The solution to this recursion gives $t(n, \kappa) \in O(n(1{+}\log{\kappa}))\subseteq O(n\log n)$. In the sequel, this algorithm will be named as {\tt{Test-And-Divide}}.

\begin{LONG}
The {\tt{Test-And-Divide}} algorithm partitions the input chain into simple subchains. If it was possible to partition the input chain into the minimum number $\kappa$ of simple subchains in linear time, then the same approach described by Barbay and Navarro~\cite{2013-TCS-OnCompressingPermutationsAndAdaptiveSorting-BarbayNavarro} could be applied to obtain a refined analysis in function of $O(n(1+\mathcal{H}(n_1, \dots, n_\kappa))) \subseteq O(n(1{+}\log{\kappa})) \subseteq O(n\log{n})$. But, as far as we know, there does not exist any linear time algorithm to accomplish this task.
\end{LONG}

In the recursion tree of the execution of the {\tt{Test-And-Divide}} algorithm on input $C$ of $n$ points, every node represents a 
subchain of $C$. The cost of every node is linear in the size of the subchain that it represents. The simplicity 
test and the merge process are both linear in the number of points in the subchain.\begin{LONG} When this subchain is simple 
the node that represents this subchain is a leaf.\end{LONG} Every time this algorithm discovers that the polygonal chain is simple, 
\begin{LONG}it executes a number of operations linear in the size of the chain and\end{LONG} the corresponding node in the recursion tree becomes a leaf.

\subsubsection{Width Analysis: A Warm-up.} 

Consider for illustration the particular case of a polygonal chain $C$ of $n = 2^m$ planar points such that $C$ can be partitioned into the minimum number of
simple subchains $\kappa = m$ of lengths $2^1, 2^1, 2^2, 2^3 \dots, 2^{m-1}$, respectively. Every time the algorithm cuts the current chain in half, the right subchain is simple and the recursive call is made only in the left subchain. Hence, the recursion tree of the algorithm on input $C$ has only two nodes per level, one of which is a leaf (see Figure~\ref{fig:rec}). The overall running time of the algorithm on $C$ is then within $O(n)$.
\begin{figure}
\centering
\includegraphics[scale=0.9]{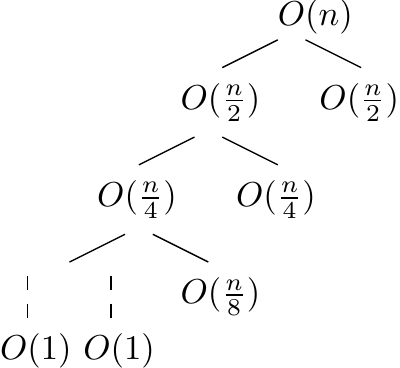}
\caption{The recursion tree of $A$ on $C$. Each node represents a recursive call. Noted in each node is the asymptotic complexities of the simplicity test and the merging process on the subchain that it represents.}
\label{fig:rec}
\end{figure}
%

\begin{definition}[Width]
The {\em{width $\omega$}} of the recursion tree in the execution of the {\tt{Test-And-Divide}} algorithm on input $C$ is the maximum number of nodes at any level.
\end{definition}

Levcopoulos et al.~\cite{2002-SWAT-AdaptiveAlgorithmsForConstructingConvexHullsAndTriangulationsOfPolygonalChains-LevcopoulosLingasMitchell} analyzed the complexity of this algorithm
in the worst case over instances of fixed size $n$ and $\kappa$. The following lemma gives an alternate analysis in the worst case over instances of fixed size $n$
and width $\omega$.

\begin{lemma}
Let $\omega$ be the width of the recursion tree in the execution of the {\tt{Test-And-Divide}} algorithm on input $C$ of $n$ planar 
points. The complexity of this algorithm on input $C$ is within $O(\omega n)$.
\end{lemma}




\subsubsection{Refined Analysis.}
\label{sec:finer}

\begin{LONG}
Let $\langle \ell_1, \dots, \ell_m \rangle$ be the vector formed by the sizes of the 
subchains represented by the $m$ leaves of the recursion tree of the {\tt{Test-And-Divide}} algorithm on input $C$ (such that $\sum_{i=1}^m{\ell_i} = n$).
The number of operations ``saved'' by the
algorithm every time it discovers a leaf of size $\ell_i$ is within $\Omega(\ell_i\log{\ell_i})$ (the cost of the subtree of the perfect binary 
tree rooted in a node of size $\ell_i$) minus $O(\ell_i)$ (the operations in the leaf are not saved). The time complexity $T(C)$ of 
this algorithm on input $C$ is within $O(n\log{n})$ (the cost of the perfect binary tree) minus $\Omega(\sum_{i=1}^m \ell_i \log{\ell_i} - \ell_i)$ (the number of 
operations saved by the algorithm). So, $T(C) \subseteq O(n\log{n} - \sum_{i=1}^m (\ell_i \log{\ell_i} - \ell_i)) = O(n(1 + {\mathcal{H}(\ell_1, \dots, \ell_m)})) \subseteq O(\omega n) \cap O(n\log{m}) \subseteq O(n\log{n})$.

The following lemma summarizes this finer analysis of the {\tt{Test-And-Divide}} algorithm in function of the width and the number of leaves in the recursion tree.

\begin{lemma}
Let $\langle \ell_1, \dots, \ell_m \rangle$ be the vector formed by the sizes of the subchains represented by the $m$ leaves of the recursion 
tree in the execution of the algorithm $A$ on input $C$ of $n = \sum_1^m{\ell_i}$ planar points. Let $\omega$ be the maximum width 
of the recursion tree. The time complexity of $A$ on $C$ is within $O(n(1 + {\mathcal{H}(\ell_1, \dots, \ell_m)})) \subseteq O(\omega n) \cap O(n\log{m}) \subseteq O(n\log{n})$.
\end{lemma}

Is there a relationship between the vector
$\langle \ell_1, \dots, \ell_m \rangle$ formed by the sizes of the subchains represented by the leaves of the recursion tree in the execution of the {\tt{Test-And-Divide}} algorithm from on input $C$
and the vector $\langle n_1, \dots, n_{\kappa} \rangle$ formed by the sizes of a partition of $C$ into $\kappa$ simple subchains?
\end{LONG}
%

For a given polygonal chain, there can be several partitions into simple subchains of minimum size $\kappa$ for it. The Levcopoulos et al.'s analysis is in the worst case over instances for $n$ and $\kappa$ fixed. We describe below a refined analysis which takes into account the relative imbalance between the sizes of the subchains. The idea behind the refinement is to bound the number of operations that the algorithm executes for every simple subchain. This analysis makes it possible to identify families of instances where the complexity of the algorithm is linear even though the number of simple subchains into which the chain is split is logarithmic.

\begin{theorem}
\label{theo:simple}
Let $\langle n_1, \dots, n_{\kappa} \rangle$ be the vector formed by the sizes of the subchains of any partition $\Pi$ of the chain $C$ into 
the minimum number $\kappa$ of simple subchains. The time complexity of the {\tt{Test-And-Divide}} algorithm on input $C$ is within
$O(n(1+\mathcal{H}(n_1, \dots, n_{\kappa}))) \subseteq O(n(1{+}\log{\kappa})) \subseteq O(n\log{n})$, which is worst-case optimal in the comparison model over instances of $n$ points that can be partitioned into $\kappa$ simple subchains of sizes $\langle n_1, \dots, n_{\kappa} \rangle$.
\end{theorem}

\begin{proof}
Fix the subchain $c_i$ of size $n_i$ in $\Pi$. In the worst case, the algorithm considers the $n_i$ points of $c_i$ for the simplicity test, and
the merging process in all the levels of the recursion tree from the first level to the level $\lceil \log{\frac{n}{n_i}} \rceil + 1$, because the sizes of the subchains in these levels are greater than $n_i$. In the next level, one of the nodes $\ell$ of the recursion tree fits completely inside $c_i$ and therefore it becomes a leaf. Hence, at least $\frac{n_i}{4}$ points from $c_i$ are dismissed for the following iterations. The remaining points of $c_i$ are in the left or the right ends of subchains represented by nodes in the same level of $\ell$ 
in the recursion tree. In all of the following levels, the number of operations of the algorithm involving points from $c_i$ can be bounded by the size of the subchains in those levels. So, the sum of the number of these operations in these levels is within $O(n_i)$.
As a result, the number of operations
of the algorithm involving points from $c_i$ is within $O(n_i\log{\frac{n}{n_i}} + n_i)$. In total, the time complexity of the algorithm is within 
$O(n + \sum_{i=1}^\kappa n_i\log{\frac{n}{n_i}}) = O(n(1+\mathcal{H}(n_1, \dots, n_\kappa))) \subseteq O(n(1{+}\log{\kappa})) \subseteq O(n\log{n})$. 

We prove the optimality of this complexity by giving a tight lower bound.
Barbay and Navarro~\cite{2013-TCS-OnCompressingPermutationsAndAdaptiveSorting-BarbayNavarro} showed a lower bound of $\Omega(n(1+{\cal H}(r_1,\ldots,r_\rho)))$ in the comparison model
for {\sc{Sorting}} a sequence of $n$ numbers, in the worst case over instances covered by $\rho$ runs (increasing or decreasing) of lengths $r_1, \dots, r_\rho$, respectively, summing to $n$. 
The {\sc{Sorting}} problem can be reduced in linear time to the problem of computing 
the {\sc{Convex Hulls}} of a chain of $n$ planar points that can be cut into 
$\rho$ simple subchains of lengths $r_1, \dots, r_\rho$, respectively. For each real number $r$, this is done by producing
a point with $(x,y)$-coordinates $(r,r^2)$. The $\rho$ runs (alternating increasing and decreasing) are transformed into $\rho$ simple subchains of the same lengths. 
The sorted sequence of the numbers can be obtained from the {\sc{Convex Hull}} of the points in linear time. \qed
\end{proof}

\begin{INUTILE}
\begin{lemma}
Let $\langle n_1, \dots, n_{\kappa} \rangle$ be the lengths of the subchains of a partition of the chain $C$ into 
the minimum number $\kappa$ of simple subchains. Let $\langle \ell_1, \dots, \ell_m \rangle$ be the vector formed by the sizes 
of the subchains represented by the leaves of the recursion tree in the execution of the algorithm $A$ on input $C$. 
Then the number $\mathcal{S} \in \Omega(\sum_1^m \ell_i \log{\ell_i})$ of operations ``saved'' by the algorithm is within $\Omega(\sum_1^\kappa n_i \log n_i)$.
\end{lemma}

\begin{proof}
The subchain of length $2\ell_1$ that starts in the first point of $C$ is not simple. So, the length $n_1$ of the first simple
subchain of $\Pi$ is less than $2\ell_1$, implying that $\ell_1 > \frac{n_1}{2}$. 
In a similar way, the chain of length $2\ell_m$ that ends in the last point of $C$ is not simple. So, the length $n_\kappa$ of the 
last simple subchain of $\Pi$ is less than $2\ell_m$, implying that $\ell_m > \frac{n_{\kappa}}{2}$. Also, $\ell_1 < n_1 + n_2$ and $\ell_m < n_{\kappa-1} + n_{\kappa}$
because otherwise the first two chains of $\Pi$ of lengths $n_1$ and $n_2$ or the last two chains of $\Pi$ of lengths $n_{\kappa-1}$ and $n_\kappa$ 
could be merged into a simple single subchain, contradicting that $\kappa$ is minimum.
If there exist $i$, $j$ such that $\ell_i > n_j + \frac{n_{j+1}}{2}$, then $\ell_i\log{\ell_i} > n_j{\log{n_j}} + \frac{n_{j+k}}{2} \log{\frac{n_{j+k}}{2}}$.
When the algorithm $A$ cuts the chain $C$ into $C'$ and $C''$, the cutting point $p$ could be a partitioning point of $\Pi$ or could lie inside an
element of $\Pi$. In the first case, $p$ is the last and the first point of two consecutive elements of $\Pi$ of lengths $n_j$ and $n_{j+1}$ such that $\sum_1^j n_i = \sum_{j+1}^\kappa n_i$.
In such a case, similar to the analysis of $\ell_1$ and $\ell_{\kappa}$ there exist $\ell_p$ and $\ell_{p+1}$ such that $\ell_p > \frac{n_j}{2}$ and $\ell_{p+1} > \frac{n_{j+1}}{2}$ and 
$l_p < n_{j-1} + n_j$ and $l_{p+1} < n_{j+1} + n_{j+2}$. In the second case, $p$ lies inside an element of $\Pi$ of length $n_j$ such that $ \sum_1^j n_i > \frac{n}{2} > \sum_1^{j-1} n_i$.
In such a case, one of the subchains of lengths $n_j'$ and $n_j''$ into which $p$ partitions the chain of lengths $n_j$ is larger than $\frac{n_j}{2}$. So, there
exists $\ell_p$ such that $\ell_p > \frac{n_j}{4}$.
As a consequence: $\sum_1^m \ell_i\log{\ell_i} > \sum_1^\kappa \frac{n_i}{4}\log{\frac{n_i}{4}} = \frac{1}{4}\sum_1^\kappa{n_i\log{n_i}}- \frac{1}{2}n$. \qed
\end{proof}
\end{INUTILE}

\subsection{Computing Delaunay Triangulations: Adaptivity to Monotone Histograms}
\label{sec:monotone}

We propose a new input-order based algorithm for computing the {\sc{Delaunay triangulation}} and the {\sc{Voronoi diagram}} of a set of $n$ planar points. A \emph{monotone histogram} is a sequence of points sorted with respect to two orthogonal directions. The algorithm takes advantage of the minimal number $\mu$ of monotone histograms and their sizes into which a polygonal chain can be cut. A monotone histogram is also a simple polygonal chain. Therefore, an algorithm for computing {\sc{Convex Hulls}} adaptive to the decomposition of the input into monotone histograms
is obtained as a corollary of Theorem~\ref{theo:simple}. We extend the refined analysis to the computation of {\sc{Delaunay triangulations}} and {\sc{Voronoi diagrams}} adaptive to the decomposition of
the input into monotone histograms.

Djidjev and Lingas~\cite{1995-IJCGA-OnComputingVoronoiDiagramsForSortedPointSets-DjidjevLingas} described an algorithm which, given a monotone histogram, computes the {\sc{Voronoi diagram}} (and hence the {\sc{Delaunay triangulation}}) of the input sequence in linear time. This algorithm suggests a way to partition the input into subsequences such that the {\sc{Delaunay triangulation}} of each subsequence can be computed in linear time in its length. The partitioning algorithm cuts the sequence into $\mu$ monotone histograms with respect to two orthogonal directions $d_1$ and $d_2$. The first two points of the subsequence determine the ordering defined by the combination of ascending and descending with respect to $d_1$ and $d_2$.
\begin{LONG}
Given a sequence of points and two orthogonal directions, it is possible to test whether the sequence is a monotone histograms with respect to these two directions in linear time.
\end{LONG}

\begin{LONG}
\subsubsection{Binary Merge of Voronoi Diagrams.}
\label{sec:vormerge}

Kirkpatrick~\cite{1979-FOCS-EfficientComputationOfContinuousSkeletons-Kirkpatrick} described 
a linear time algorithm for the merging of two arbitrary {\sc{Voronoi diagrams}}. 
Given the {\sc{Voronoi diagrams}} of two disjoint point sets $P$ and $Q$, the algorithm finds the {\sc{Voronoi diagram}} of 
$P \cup Q$ in time within $O(|P| + |Q|)$. 
The plane is partitioned into points closer to $P$,
points closer to $Q$, and points equidistant from $P$ and $Q$.
The points equidistant from $P$ and $Q$ are defined as the \emph{contour}
separating $P$ and $Q$. The \emph{contour} is composed of straight line segments: it is formed
from the edges of the {\sc{Voronoi diagram}} of $P \cup Q$ that separates the points in $P$ from the points in $Q$. Inside the region of points closer to $P$ (resp. $Q$) the {\sc{Voronoi diagram}} of $P \cup Q$ and the {\sc{Voronoi diagram}} of $P$
(resp. $Q$) are identical. Thus, the merging of two {\sc{Voronoi diagrams}} can be seen as the process of cutting the {\sc{Voronoi 
diagrams}} of $P$ and $Q$ along the contour.

This leads to a divide-and-conquer algorithm for constructing 
the {\sc{Voronoi diagram}} of a set of $n$ points and hence for computing the {\sc{Delaunay triangulation}}
in time within $O(n\log n)$.
This time complexity of $O(n\log n)$ is asymptotically optimal in the comparison model in the worst case over instances composed of $n$ points. Shamos and Hoey~\cite{1975-FOCS-ClosestPointProblems-ShamosHoey} showed that the construction of any triangulation over $n$ points requires $\Omega(n\log n)$, as {\sc{Sorting}} can be reduced to computing the triangulation of $n+1$ points, which yields an asymptotic computational lower bound of $\Omega(n\log n)$ in the worst case over sets of $n$ planar points, in the comparison model.
\end{LONG}

\subsubsection{Multiary Merge.}
\label{sec:multiary-merge}

\begin{SHORT}
Kirkpatrick~\cite{1979-FOCS-EfficientComputationOfContinuousSkeletons-Kirkpatrick} described 
a linear time algorithm for the merging of two arbitrary {\sc{Voronoi diagrams}}.
This leads to a divide-and-conquer algorithm for constructing 
the {\sc{Voronoi diagram}} of a set of $n$ points and hence for computing the {\sc{Delaunay triangulation}}
in time within $O(n\log n)$.
\end{SHORT}

Given $\mu$ {\sc{Delaunay triangulations}} of sizes $v_1, \dots, v_{\mu}$, respectively, to be merged, we make a sequence of binary merges, reducing at each step the number of {\sc{Delaunay triangulations}} to be merged by 1. 
The merging process can be represented by a binary tree where the internal nodes are the merged {\sc{Delaunay triangulations}}, and the leaves are the original 
$\mu$ {\sc{Delaunay triangulations}}. Merging the two {\sc{Delaunay triangulations}} of minimum sizes at each step, similar as the Huffman code algorithm~\cite{1952-IRE-AMethodForTheInstructionOfMinimumRedundancyCodes-Huffman}, further improves the merging process, which takes advantage of the potential disequilibrium in the distribution of the points between the $\mu$ {\sc{Delaunay triangulations}}. 
\begin{LONG}
We can apply the Huffman~\cite{1952-IRE-AMethodForTheInstructionOfMinimumRedundancyCodes-Huffman} algorithm to the vector 
$\langle v_1, \dots, v_{\mu} \rangle$, thus obtaining a Huffman-shaped tree representing the merging process.
\end{LONG}

\begin{LONG}
\begin{lemma}[Multiary Merge]\label{lem:multiaryMerge}
Given $\mu$ {\sc{Delaunay triangulations}} of respective sizes $\langle v_1,\ldots,v_\mu \rangle$ summing to $n=\sum_{i=1}^\mu v_i$, there is an algorithm computing the {\sc{Delaunay triangulation}} of the $n$ points in time within $O(n(1+{\cal H}(v_1, \dots, v_\mu)))\subseteq O(n(1+\log\mu)) \subseteq O(n\log{n})$.
\end{lemma}
\end{LONG}

\begin{LONG}
\begin{proof}
The algorithm follows the same steps as the algorithm suggested by Huffman~\cite{1952-IRE-AMethodForTheInstructionOfMinimumRedundancyCodes-Huffman} on a set of $\rho$ messages of probabilities $\{r_i/n\}_{i\in[1..\rho]}$:
\begin{enumerate}
\item Initialize a heap $H$ with the $\rho$ {\sc{Voronoi diagrams}}, indexed by their size;
\item While $H$ contains more than one {\sc{Voronoi diagram}}
  \begin{itemize}
\item extract the two smallest {\sc{Voronoi diagrams}} from $H$, of respective sizes~$n_1$ and~$n_2$,
\item merge them into a {\sc{Voronoi diagram}} $T$ of size $n_1 + n_2$, and
\item insert $T$ in $H$.
  \end{itemize}
  \end{enumerate}

This algorithm executes in time within $O(n(1+{\cal H}(v_1,\ldots,v_\mu)))$.
The merging process is then represented by a tree with the shape of a Huffman~\cite{1952-IRE-AMethodForTheInstructionOfMinimumRedundancyCodes-Huffman}
tree. Consider the $i$-th {\sc{Voronoi diagram}} of the input $\forall i\in[1..\mu]$: let $c_i$ be the binary string describing the path leading from the root to the corresponding leaf, and $l_i$ the length of this path. The sum of the computational costs of the binary merges is the sum of the sizes of the {\sc{Voronoi diagram}} computed. The $i$-th {\sc{Voronoi diagram}} contributes a cost within $O(v_i)$ to $l_i$ levels, which has a sum within $O(\sum_{i=1}^\mu l_i v_i)$.
Consider the binary tree where the $\mu$ initial {\sc{Voronoi diagrams}} are leaves and the $\mu-1$ computed {\sc{Voronoi diagrams}} are internal nodes:
\begin{itemize}
\item The set of binary strings $\{c_1,\ldots,c_\mu\}$ is a prefix free code, i.e. no code is prefix of another root-to-leaf path, simply because they are paths in a tree.
\item The lengths of those codes minimize $\sum_{i=1}^\mu l_i r_i$ as a property of Huffman~\cite{1952-IRE-AMethodForTheInstructionOfMinimumRedundancyCodes-Huffman} codes.
\end{itemize}
By the optimality of Huffman codes,  this complexity is within a linear term of the entropy of the distribution $(v_1,\ldots,v_\mu)$, i.e. $\sum_{i=1}^\mu l_i r_i\in O(n(1+{\cal H}(v_1,\ldots,v_\mu)))$.
This yields the final time complexity, within $O(n(1+{\cal H}(v_1,\ldots,v_\mu)))$. \qed
\end{proof}
\end{LONG}

The combination of the partitioning algorithm and the merging process yields an optimal algorithm computing these structures 
adaptive to the decomposition of the input into monotone histograms.

\begin{theorem}
Let $d_1$ and $d_2$ be two perpendicular directions. Let $S$ be a sequence of $n$ planar points. Let $\mu$ and $\langle v_1,\ldots,v_\mu \rangle$ be the minimum number of monotone histograms with respect to $d_1$ and $d_2$ and the sizes of these monotone histograms, respectively, in which $S$ can be cut. The {\sc{Delaunay triangulation}} and the {\sc{Voronoi diagram}} of $S$ can be computed in time within $O(n(1+{\cal H}(v_1,\ldots,v_\mu))) \subseteq O(n(1{+}\log{\mu})) \subseteq O(n\log{n})$, which is worst-case optimal in the comparison model over instances of $n$ points that can be partitioned into monotone histograms of sizes $\langle v_1, \dots, v_{\mu} \rangle$.
\end{theorem}

\begin{LONG}
\begin{proof}
The combination of the partitioning algorithm and the merging process yields an algorithm computing these structures within this time.
In order to provide a lower bound, we use again the result of $\Omega(n(1+{\cal H}(r_1,\ldots,r_\rho)))$ for {\sc{Sorting}} a sequence of $n$ numbers, in the worst case over instances
covered by $\rho$ runs of lengths $r_1, \dots, r_\rho$ summing $n$ in the comparison model, demonstrated by Barbay and Navarro~\cite{2013-TCS-OnCompressingPermutationsAndAdaptiveSorting-BarbayNavarro}.
This problem can
be reduced in linear time to the problem of computing the {\sc{Delaunay triangulation}} of a sequence of $n$ planar points covered by $\rho$ monotone histograms of lengths $r_1, \dots, r_\rho$
with respect to the coordinates axes. 
The $\rho$ runs are transformed into $\rho$ monotone histograms of the same lengths, using points on the parabola, in linear time. The sorted sequence of the numbers can be obtained from the {\sc{Delaunay triangulation}} of the points in linear time. \qed
\end{proof}
\end{LONG}

\begin{SHORT}
The proof is similar to the one of Theorem~\ref{theo:simple}, we defer it to the extended version~\cite{2015-ARXIV-RefiningTheAnalysisOfDivideAndConquer-BarbayOchoaPerez}.
\end{SHORT}
%

Such examples of complete refinements are not the rule: in the following section we describe some examples where such refinements are problematic or impossible.

\section{Partial Refinements}
\label{sec:when}

We describe a new partitioning algorithm for a sequence of points \begin{LONG}in Sections~\ref{sec:incremental} to~\ref{sec:optimal}\end{LONG} such that the {\sc{Delaunay triangulation}}
of each subsequence can be computed in linear time in its length, and discuss alternate partitions\begin{LONG} in Section~\ref{sec:partitioning}\end{LONG}.
Each different partitioning algorithm, in combination with the merging process described in Section~\ref{sec:multiary-merge}, yields a different algorithm adaptive to the input order.

\begin{LONG}
\subsection{Incremental Construction}
\label{sec:incremental}
\end{LONG}

\begin{LONG}
Many Computational Geometry algorithms use tests known as the orientation and incircle tests~\cite{1997-DCG-AdaptivePrecisionFloatingPointArithmeticAndFastRobustGeometricPredicates-Shewchuk}. 
The \emph{orientation test} determines whether a point lies to the left of, to the right of, or on a line or 
plane defined by other points. The \emph{incircle test} determines whether a point lies inside, outside, 
or on a circle defined by other points.
\end{LONG}

Green and Sibson~\cite{1977-TCJ-ComputingDirichletTessellationsInThePlane-GreenSibson} proposed 
the first incremental algorithm for computing the {\sc{Delaunay triangulation}} of a point set which finds the 
triangle containing each new point, and updates the diagram by correcting the edges violating the 
circumcircle condition\begin{LONG} ---an operation named \emph{flipping}\end{LONG}.
\begin{LONG}
The algorithm adds points to the structure one by one. For each point, it performs two basic steps:

\begin{enumerate}
\item The algorithm finds the triangle containing the new point using the structure as a guide to the relative position of the points.
A greedy approach for locating the point is to start at an edge in the structure and to walk across adjacent edges in the 
direction of the new point until the correct triangle is found. Orientation tests~\cite{1997-DCG-AdaptivePrecisionFloatingPointArithmeticAndFastRobustGeometricPredicates-Shewchuk}
are performed on each edge of such a path to see whether the new point lies on the correct side of that edge. We call the operations
involved in this walk \emph{navigation} operations. 

\item It then updates the structure adding three new edges
from the point inserted to the vertices of the triangle containing the point and flips all invalid edges 
resulting from the insertion. Note that flipping an edge can make another edge invalid, but that each edge
is flipped at most once, so each insertion can trigger at most a linear number of flips.
\end{enumerate}
\end{LONG}
The time complexity of this incremental algorithm for computing the {\sc{Delaunay triangulation}} is within $O(dn)\subseteq O(n^2)$, where $d\in [1..n]$ is an upper bound on the amount of operations required to insert each point and to correct the {\sc{Delaunay triangulation}} for the instance being considered. We show the problems that arise when we try to adapt this algorithm in order to obtain a new one whose time complexity is within $O(n\log{d}) \subseteq O(n\log{n})$.

\begin{LONG}
\subsection{Adaptivity to D-linear Runs}
\label{sec:adaptivity}

Consider an incremental algorithm $G$ computing the {\sc{Delaunay triangulation}} of $n$ planar 
points in time within $O(n)$ in the best case; and an algorithm $M$ merging two {\sc{Delaunay triangulations}}
of sizes summing to $n$ in time within $O(n)$ in the worst case.
Given a sequence $S$ of $n$ distinct planar points and a constant $k$, the following naive algorithm computes 
its {\sc{Delaunay triangulation}} in time within $O(n)$ in the best case, in time within $O(n\log n)$ in the worst case, 
and in time within $O(n(1+\mathcal{H}(r_1, \dots, r_\rho))) \subseteq O(n\log\rho)\subseteq O(n\log n)$ in the general case,
where $\rho\in[1..n]$ and $r_1, \dots, r_\rho$ (such that $n=\sum_{i=1}^\rho r_i$) measure the difficulty of partitioning the instance into ``easy'' subinstances:

\begin{enumerate}
\item Run the incremental algorithm $G$ on $S$ until, for the $i$-th point $p$, it performs either more 
than $k$ operations to locate the point on the triangulation, or more than $k$ flip operations on it.
\item Store the {\sc{Delaunay triangulation}} of the $i-1$ first points, and restart the greedy incremental algorithm $G$ 
on the sequence formed by $p$ and its $n-i$ successors in the input sequence $S$, until no points are left in~$S$.
\item Let $\rho\in[1..n]$ and $r_1, \dots, r_\rho$ (such that $n=\sum_{i=1}^\rho r_i$) be the number of {\sc{Delaunay triangulations}} computed in this way and the sizes
of these {\sc{Delaunay triangulations}}, respectively.  Merge the $\rho$ {\sc{Delaunay triangulations}}, in overall time within $O(n(1+\mathcal{H}(r_1, \dots, r_\rho))) \subseteq O(n(1{+}\log\rho)) \subseteq O(n\log{n})$.
\end{enumerate}
\end{LONG}

\begin{LONG}
\subsubsection{Greedy Partitioning.}
\label{sec:greedy-partitioning}
The algorithm described above suggests a way to partition the input into subsequences such that the {\sc{Delaunay triangulation}}
of each subsequence can be computed in linear time in its length.

We use the algorithm described to define ``easy'' sequences of points for the computation of the {\sc{Delaunay triangulation}}:
\end{LONG}


\begin{LONG}
\begin{definition}[D-linear]
Consider a sequence $S$ of $n$ distinct planar points. Given an integer value $k>0$, $S$ is \emph{$k$-D-linear} (for ``Delaunay linear'') if $G$ performs at most $k$ location and flip operations for each point while computing the {\sc{Delaunay triangulation}} of $S$.  If $S$ is $k$-D-linear and $k\in O(1)$ is a constant independent of $n$, we say that $S$ is \emph{D-linear}.
\end{definition}

Such a simple definition yields a simple partitioning of the input sequence into subsequences of consecutive positions such that the {\sc{Delaunay triangulation}} of each subsequence can be computed in linear time.
\begin{definition}[D-linear Run]
A \emph{D-linear run} in a sequence $S$ of points is a \emph{D-linear} subsequence formed by consecutive points in $S$.
\end{definition}
\end{LONG}
\begin{SHORT}
Consider a sequence $S$ of $n$ planar points. Given an integer value $k>0$, the partitioning algorithm greedily adds points to the \emph{run} while the incremental algorithm $G$ has not executed a number of operations exceeding the threshold $k$, i.e. the $i$-th point $p$ is added to the current \emph{run} if $G$ performs at most $k$ operations to locate $p$ in the triangulation and at most $k$ update operations on it. This algorithm yields a partition of the sequence $S$ into subsequences such that the {\sc{Delaunay triangulation}} of each subsequence can be computed in linear time.
\end{SHORT}
\begin{LONG}
\subsection{Non Optimality of D-linear Runs Partitioning}
\label{sec:optimal}

\begin{INUTILE}
Given a sequence $S$, an optimal partition of $S$ into D-linear runs is a partition with the minimum number of D-linear runs.
We show in this section that the greedy partitioning algorithm into D-linear runs does not always yield
a partition into the minimum possible number of D-linear runs.
\end{INUTILE}
\end{LONG}

\begin{SHORT}
  This greedy partitioning algorithm does not always yield a partition into the minimum possible number of \emph{runs}. It is possible to describe a family of sequences where for each sequence $S_n$ of size $n$ of the family the gap between the number of \emph{runs} yielded by this algorithm in $S_n$ and the minimum number of \emph{runs} into which $S_n$ can be partitioned is $\Theta(n)$.

  A sequence $S$ of $n$ points can be partitioned into a minimum number of \emph{runs} in time within $O(n^2)$: suffices to compute for each point $p$ in $S$ the largest \emph{run} starting at $p$.
  The running time of the partitioning algorithm must be within $O(n\log{n})$ but adaptive to the same parameter as the merging in order to 
obtain an adaptive algorithm computing the {\sc{Delaunay triangulation}} in time within $o(n\log n)$ on some classes of instances.
\end{SHORT}

\begin{LONG}
We describe a family of sequences $\mathcal{S}$, such that for all positive integer $n$, 
there is a sequence $S_n$ of $n$ points in $\mathcal{S}$ where the greedy partitioning algorithm yields three D-linear 
runs when $S_n$ can be optimally partitioned into just two D-linear runs. 

The {\sc{Delaunay triangulation}} of $m\ge 2$ points on the parabola $y=x^2$ is such that the leftmost point is adjacent to all the other $m-1$ points, and the left-to-right order is a D-linear run. The proof of the following lemma is based on this fact.

\begin{lemma}\label{lem:2-3}
For all positive integers $n$, there exists a sequence of $n$ planar points where the greedy partitioning algorithm yields three D-linear runs, whereas the optimal partition of this sequence has only two D-linear runs. 
\end{lemma}

\begin{proof}
Let $k$ be the threshold used in the definition of D-linear sequences. Let $A$ be a set of $m>k$ points  in the parabola $y=x^2$, 
denoted $p_1,\ldots,p_m$ from left to right. 
Let $C$ be the circumcircle of the triangle $p_1,p_2,p_3$.
Let $u,v,w$ be three points such that the triangle $\Delta uvw$ with vertex set $\{u,v,w\}$
is small enough with respect to the convex hull of $A$, and located inside $C$.
Let $A'=\{p'_1,\ldots,p'_m\}$ be a scaled copy of $A$, located inside $\Delta uvw$.
Let $B$ be the $k+3$ vertex set of $k+1$ decreasing area and disjoint triangles such as every two consecutive triangles
share an edge (see Fig.~\ref{fig:counter}). 
Suppose that $B$ is ordered such that the sequence of points is a D-linear run.
The points $u,v$ and $w$ are the last three points of $B$. The triangle formed by the first three points of $B$ contains the point $p_1$ of $A$.
Let $P=A \cup A' \cup B$ be ordered as $\langle B,p_1,\ldots, p_m,p'_1,\ldots, p'_m \rangle$.
Applying the greedy partitioning algorithm to $P$ gives the three D-linear runs $\langle B \rangle$, $\langle p_1,\ldots, p_m \rangle$, 
and $\langle p'_1,\ldots, p'_m \rangle$ since: (1) adding $p_1$ to the {\sc{Delaunay triangulation}} of $B$ requires
traversing $k$ triangles to locate $p_1$; and (2) adding $p'_1$ to the {\sc{Delaunay triangulation}} of 
$A$ requires creating $m$ triangles (i.e.\ $p'_1$ is adjacent to every point in 
$A$). However, $\langle B\setminus \{u,v,w\} \rangle$ is a D-linear run and $\langle u,v,w,p_1,\ldots, p_m,p'_1,\ldots, p'_m \rangle$ is another one. Note that the triangle $\Delta uvw$ blocks the points in $A'$ of being adjacent 
to any point in $A$. \qed
\end{proof}

\begin{figure}[t!]
\centering
\includegraphics[scale=0.9]{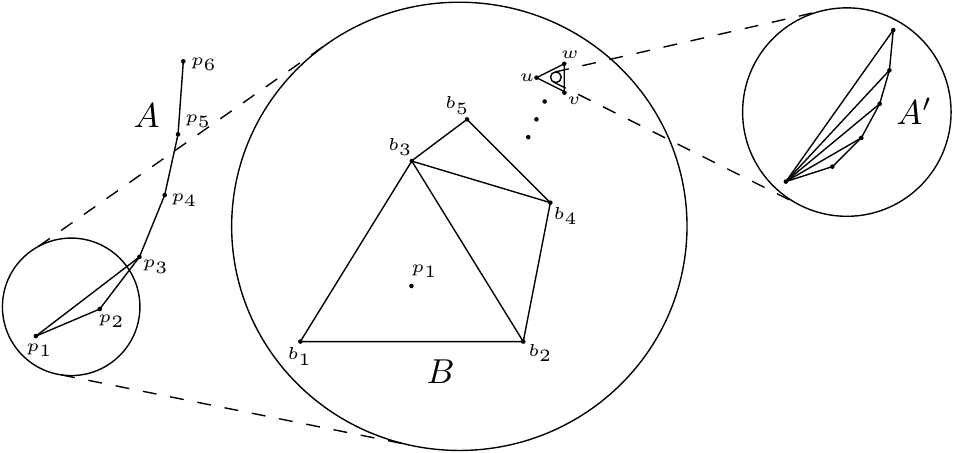}
\caption{The sequences $A$, $B$ and $A'$.}
\label{fig:counter}
\end{figure}

This family of sequences $\mathcal{S}$ can be generalized to show a gap of $\Theta(n)$ between the number of D-linear runs yielded by the greedy partitioning algorithm in $S_n$ and the minimum number of D-linear runs in which $S_n$ can be partitioned.
\end{LONG}

\begin{LONG}
\begin{lemma}\label{lem:greedy}
For all positive integer $n$, there exists a sequence of $n$ planar points where the greedy partitioning 
algorithm yields $\rho \in \Theta(n)$ D-linear runs, and the optimal partition of this sequence has only 2 D-linear runs.
\end{lemma}
\end{LONG}

\begin{LONG}
\begin{proof}
Let $A$ be a set of $m > k$ consecutive points in the parabola $y=x^2$ denoted 
$p_1,\ldots,p_m$ from left to right.
Let $C$ be the circumcircle of the triangle $p_1,p_2,p_3$.
Let $u,v,w$ be three points such that the triangle $\Delta uvw$ with vertex set $\{u,v,w\}$
is small enough with respect to the convex hull of $A$, and located inside $C$.
Let $A'$ be a D-linear run sequence of $l$ points, of increasing \emph{y-coordinate} and decreasing \emph{x-coordinate}, 
inside $C$ denoted $p'_1,\ldots,p'_l$. During the incremental construction 
of the {\sc{Delaunay triangulation}} of the sequence $AA'$ the first point of $A'$ forces a flipping of $m$
edges in $A$ (i.e.\ $m$ new edges are created connecting $p'_1$ with all points in $A$). Every 
next point of $A'$ forces a flipping of the $m$ edges connecting points in $A$ with points in $A'$.
Let $q'_1$ and $q'_2$ be two consecutive points from the sequence $A'$ and let $q_1,\dots,q_{k+1}$ be 
$k+1$ consecutive points from $A$. The sequence $Q=\langle q'_1,q_1,\ldots, q_{k+1},q'_2 \rangle$ is not a D-linear run because 
$q'_2$ forces a flipping of the $k+1$ edges connecting $q'_1$ with all the points in $q_1,q_2,\ldots,q_{k+1}$.
%
Let $B$ be the ($k+3$) vertex sequence of $k+1$ decreasing area and disjoint triangles so that every two consecutive triangles
share an edge. The points $u,v$ and $w$ are the last three points of $B$. The triangle formed by the first three points of $B$ contains inside the point $p_1$ of $A$.
$B$ is ordered such that the sequence of vertices is a D-linear run.
Let $P=A \cup A' \cup B$ be ordered as $\langle B,p_1,\ldots, p_{k+1},p'_1,p_{k+2},\ldots, p_{2k+3},p'_2,\ldots, p_m \rangle$, 
that is, first the points of $B$, then alternate $k+1$ points of $A$ with 1 point of $A'$,
resulting subsequences similar to $Q$.
The greedy partition algorithm divides $P$ into $\min(l, \frac{m}{k+1}) + 1$ D-linear runs: $\langle B \rangle$, 
$\langle p_1,\ldots, p_{k+1} \rangle$ and every 
subsequence starting with 1 point from $A'$ followed by $k+1$ points from $A$. Since: (i) adding $p_1$ to the {\sc{Delaunay 
triangulation}} of $B$ requires traversing $k$ triangles to locate $p_1$; (ii) adding $p'_1$ to the {\sc{Delaunay triangulation}} of 
$p_1\ldots p_{k+1}$ requires creating $k$ triangles (i.e.\ $p'_1$ is adjacent to every point in $p_1\ldots p_{k+1}$); 
and (iii) adding $p'_{i+1}$ to $\langle p'_i,p_{ik+j},\ldots, p_{(i+1)k+j+1} \rangle$ produces a sequence similar to $Q$, and hence
is not a D-linear run. However, 
$\langle B\setminus \{u,v,w\} \rangle$ is a D-linear run and 
$\langle u,v,w,p_1,\ldots, p_k,p'_1,p_{k+1}, p_{k+2},\ldots,p_{2k},p'_2,\ldots, p_m \rangle$ 
is another one. Note that the triangle $\Delta uvw$ blocks the points in $A'$ of being adjacent to any vertex in $A$.
It is possible to build a sequence $l=\frac{m}{k+1}=\frac{n-k+3}{k+2}$ (1 point in $A'$ for each sequence of 
$k+1$ points in $A$) such that the number of D-linear runs yielded by the greedy partitioning algorithm is $\frac{n-k+3}{k+2} + 1$
and the optimal partition has only 2 D-linear runs. \qed
\end{proof}
\end{LONG}

\begin{INUTILE}
of respective sizes $n_1,\ldots,n_\rho$ summing to $n$ such that 
${\cal H}(n_1,\ldots,n_\rho)$ is within a constant additive term of the minimal entropy among partitions in runs.
\end{INUTILE}

\begin{SHORT}
Since the greedy partitioning algorithm does not yield an optimal partition, 
we explore more sophisticated partition techniques and show that they suffer similar problems.
\end{SHORT}

\begin{LONG}
\subsection{Other Partitioning Schemes}
\label{sec:partitioning}

Since the greedy partitioning algorithm does not yield an optimal partition, 
we explore more sophisticated partition techniques and show that they suffer similar problems.

\subsubsection{Test and Divide.}
\end{LONG}

We adapt the Levcopoulos et al.~\cite{2002-SWAT-AdaptiveAlgorithmsForConstructingConvexHullsAndTriangulationsOfPolygonalChains-LevcopoulosLingasMitchell} partition technique to cut a sequence of planar points\begin{LONG} into D-linear runs\end{LONG}.
\begin{LONG}
The \emph{test and divide} partitioning is another partition technique based on the incremental algorithm $G$ for the computation 
of the {\sc{Delaunay triangulation}} (seen in Sect.~\ref{sec:incremental}).\end{LONG}
Given a sequence $S = \langle p_1,\dots, p_n \rangle$ of $n$ planar points, the \emph{test and divide} partitioning first tests
whether the sequence $S$ is a \begin{LONG}\emph{D-linear}\end{LONG} \emph{run}. In such a case, it identifies the sequence $S$ as a \begin{LONG}\emph{D-linear}\end{LONG} \emph{run}. If not, it partitions the
sequence $S$ into $S' = \langle p_1, \dots, p_{\lfloor n/2 \rfloor} \rangle$ and $S'' = \langle p_{\lfloor n/2 \rfloor + 1}, \dots, p_n \rangle$ and 
recursively the same procedure is applied to $S'$ and $S''$. While the greedy partitioning has a linear time complexity, this
partitioning has a time complexity within $O(n\log{\rho}) \subseteq O(n\log{n})$ where $\rho$ is the number of \begin{LONG}\emph{D-linear}\end{LONG} \emph{runs} 
identified by the partitioning process.
\begin{INUTILE}
Given a sequence $S = \langle p_1,\dots, p_n \rangle$ of $n$ planar points the \emph{test-and-divide} partitioning first tests
whether the sequence $S$ is a \emph{run}. In such a case, it identifies the sequence $S$ as a \emph{run}. If not, it partitions the
sequence $S$ into $S' = \langle p_1, \dots, p_{\lfloor n/2 \rfloor} \rangle$ and $S'' = \langle p_{\lfloor n/2 \rfloor + 1}, \dots, p_n \rangle$ and 
recursively the same procedure is applied to $S'$ and $S''$. While the greedy partitioning has a linear time complexity, this
partitioning has a time complexity within $O(n\log{\rho}) \subseteq O(n\log{n})$ where $\rho$ is the number of \emph{runs} 
identified by the partitioning process.
The \emph{test-and-divide} partitioning suffers the same problems as the greedy partitioning.
\end{INUTILE}

\begin{INUTILE}
There exists a sequence $S$ of $n$ points where the \emph{test and divide} partitioning yields a number of D-linear runs linear on $n$,
and $S$ can be partitioned optimally in just 2 D-linear runs.
\end{INUTILE}

\begin{LONG}
The following lemma shows that the \emph{test and divide} partitioning suffers the same problems as the greedy partitioning.

\begin{lemma}\label{lem:bisec}
For all positive integer $n$, there exists a sequence of $n$ planar points where the test and divide partitioning 
yields $\rho \in \Theta(n)$ D-linear runs and the optimal partition of this sequence has only 2 D-linear runs. 
\end{lemma}

\begin{proof}
Consider the sequence $S$ used in the proof of Lemma~\ref{lem:greedy} where the number of points in $B$ is changed to $\lfloor n/2 \rfloor$ instead of $k+3$ then the \emph{test and divide} partitioning
will cut $S$ in the last point of $B$. $B$ is a D-linear run, but the bisection partitioning will cut the rest of $S$ in $\frac{n}{2(k+2)}$
D-linear runs. \qed
\end{proof}

\subsubsection{Amortized Greedy Partitioning.}
\end{LONG}

The greedy partitioning algorithm adds the point $p_i$ to the current \emph{run} whether the number of location and update operations
of the incremental algorithm $G$ in $p_i$ is at most a threshold $k$. But it is possible that in some points the 
number $r$ of operations would be much lower than $k$. The \emph{amortized greedy} partitioning algorithm makes use of the $k-r$ remaining operations in the subsequent points, similar to the accounting method for
amortized analysis. The number $k-r$ of remaining operations are credited to be used later, so that the point $p_i$ will be added to the current run if the number of navigation and update operations of $G$ in $p_i$ is less than $k$ plus this credit.
\begin{SHORT}
This partitioning algorithm suffers the same problems mentioned above. 
\end{SHORT}

\begin{LONG}
We use this algorithm to give another definition of ``easy'' sequences of points for the computation of the {\sc{Delaunay triangulation}}:

\begin{definition}[Amortized D-linear]
Consider a sequence $S = \langle p_1, \ldots, p_n \rangle$ of $n$ distinct planar points. Given an integer value $k>0$, $S$ is \emph{Amortized $k$-D-linear} if for each point $p_i \in S$, the incremental algorithm $G$ performs at most $k*i$ navigation and flip operations while computing the {\sc{Delaunay triangulation}} of the sequence $p_1, \dots, p_i$.  If $S$ is Amortized $k$-D-linear and $k\in O(1)$ is a constant independent of $n$, we say that $S$ is 
\emph{Amortized D-linear}.
\end{definition}

Such definition yields a partitioning of the input sequence into subsequence of consecutive positions:

\begin{definition}[Amortized D-linear Run]
An \emph{Amortized D-linear run} in a sequence $S$ of points is an \emph{Amortized D-linear} subsequence formed by consecutive points in $S$.
\end{definition}

We define the optimal partition into Amortized D-Linear runs as the partition with the minimum number of Amortized D-linear runs.
\end{LONG}

\begin{LONG}
\begin{lemma}
\label{lem:amortized}
For all positive integer $n$, there exists a sequence of $n$ planar points where the amortized greedy partitioning yields $\rho \in \Theta(n)$ Amortized D-linear runs and the optimal partition has only 2 Amortized D-linear runs.
\end{lemma}
\end{LONG}

\begin{SHORT}
The combination of the \emph{test and divide} and the \emph{amortized greedy} also suffers the same problem. The constructions of the families of sequences that prove the non-optimality of these four different partitioning algorithms are deferred to the extended version~\cite{2015-ARXIV-RefiningTheAnalysisOfDivideAndConquer-BarbayOchoaPerez}.
\end{SHORT}

\begin{LONG}
\begin{proof}
The construction of the sequence follows the same scheme of previous constructions (see the proof of Lemma~\ref{lem:greedy}). 
The sequence $B$ is such that the incremental
algorithm $G$ performs $k$ navigation and flip operations in almost every point in $B$, so that $B$ is still a run but the credit is almost 
zero at the end of $B$. Again, the last 3 points of $B$ form a triangle that blocks the points in $A'$ of being adjacent to any vertex in $A$.
The sequence $S$ alternates points in $A$ with points in $A'$ such that each new point in $A'$ is adjacent to every point in $A$, thus
when the numbers of points in $A$ is large enough to exceed the credit, the Amortized D-linear run is cut. This partition algorithm
yields close to $n/k^2$ Amortized D-linear runs. This sequence can be partitioned optimally in just 2 Amortized D-linear runs. 
Note that the triangle $\Delta uvw$ blocks the points in $A'$ of being adjacent to any point in $A$. \qed
\end{proof}

\subsubsection{Amortized Test and Divide Partitioning.}

The \emph{amortized test and divide partitioning} is a combination of the \emph{test and divide} partitioning and the \emph{amortized greedy}
partitioning. Given a sequence $S = \langle p_1, \dots, p_n \rangle$ of $n$ planar points we first test
whether the sequence $S$ is Amortized D-linear. In such a case, we identify the sequence $S$ as an Amortized D-linear run. 
If not, we partition the sequence $S$ into $S' = \langle p_1, \dots, p_{\lfloor n/2 \rfloor} \rangle$ and 
$S'' = \langle p_{\lfloor n/2 \rfloor + 1}, \dots, p_n \rangle$ and recursively the same procedure is applied to $S'$ and $S''$.

\begin{lemma}
For all positive integer $n$, there exists a sequence of $n$ planar points where the amortized test and divide partitioning yields 
$\rho \in \Theta(n)$ Amortized D-linear runs and the optimal partition has only 2 Amortized D-linear runs.
\end{lemma}
\end{LONG}

\begin{LONG}
\begin{proof}
The construction of the sequence follows the same scheme of previous constructions (see the proof of Lemma~\ref{lem:greedy}). 
The sequence $B$ is such that the incremental
algorithm $G$ performs $k$ navigation and flip operations in almost every point in $B$, so that $B$ is still a run but the credit is almost 
zero at the end of $B$. Again, the last 3 points of $B$ form a triangle that blocks the points in $A'$ of being adjacent to any vertex in $A$.
The sequence $S$ alternates points in $A$ with points in $A'$ such that each new point in $A'$ is adjacent to every point in $A$, thus
when the numbers of points in $A$ is large enough to exceed the credit, the Amortized D-linear run is cut. This partition algorithm
yields close to $n/k^2$ Amortized D-linear runs. This sequence can be partitioned optimally in just 2 Amortized D-linear runs. 
Note that the triangle $\Delta uvw$ blocks the points in $A'$ of being adjacent to any point in $A$. \qed
\end{proof}
\end{LONG}


\section{Discussion}
\label{sec:discussion}

We describe one technique to refine the analysis of the divide-and-conquer polynomial evaluation algorithm, and two techniques to refine the analysis of divide-and-conquer input-order adaptive algorithms for computing the {\sc{Convex Hull}}, {\sc{Voronoi Diagram}} and {\sc{Delaunay Triangulation}} of $n$ planar points in time within $O(n(1+\mathcal{H}(n_1, \dots, n_k))) \subseteq O(n(1{+}\log{k})) \subseteq O(n\log{n})$ for some parameters $n_1, \dots, n_k$ summing to $n$, where $\mathcal{H}(n_1, \dots, n_k) = \sum_{i=1}^k{\frac{n_i}{n}}\log{\frac{n}{n_i}}$ measures the ``difficulty'' of the instance. We show the optimality of the algorithms for computing {\sc{Convex Hulls}}, {\sc{Voronoi Diagrams}} and {\sc{Delaunay Triangulations}} by providing lower bounds that match the refined analyses. The later results improve on those from Levcopoulos et al.~\cite{2002-SWAT-AdaptiveAlgorithmsForConstructingConvexHullsAndTriangulationsOfPolygonalChains-LevcopoulosLingasMitchell} in that this analysis takes advantage of the sizes of the simple polygonal chains into which the input polygonal chain can be partitioned.

There is still work to be done in the directions initiated by our work. One of these techniques is based on the partitioning of a sequence of $n$ planar points into subsequences, whose corresponding {\sc{Voronoi diagrams}} and {\sc{Delaunay triangulations}} can be computed quickly, and a technique to efficiently merge those structures into a single one, inspired by Huffman~\cite{1952-IRE-AMethodForTheInstructionOfMinimumRedundancyCodes-Huffman} codes. Each combination of partition and merging techniques yields adaptive algorithms to the input order for computing \textsc{Voronoi Diagrams} and \textsc{Delaunay Triangulations}. Those results show that the \textsc{Delaunay Triangulation} can be computed in time within $o(n\log{n})$ on some classes of instances, yet we described various partitioning algorithms where the time complexity achieved is sub-optimal. The next step is to find better partitioning algorithms taking full advantage of the input order.

We classify the various refined analysis between those focusing on the structure of the instance versus those focusing on the order in which the input is given. We have been analyzing an algorithm for the \textsc{Union Set} problem that can be adapted for solving the \textsc{Sorting} problem taking advantage of the structure of the instance and the order of the input in synergy. We will continue studying potential synergistic solutions to the \textsc{Sorting} problem and their extension to the \textsc{Convex Hull} and \textsc{Delaunay triangulation} computation.


\bibliographystyle{splncs}
\bibliography{addedForThePaper,bibliographyDatabaseJeremyBarbay,publicationsJeremyBarbay}

\end{document}